\documentclass[11pt,reqno]{amsart}

\usepackage{amsmath, amsthm, amstext, amssymb, amsfonts, graphicx, bm, subfigure}
\newcommand{\Cov}{\text{Cov}}
\newcommand{\E}{\text{E}}
\newcommand{\PR}{\text{P}}
\newcommand{\overlim}[1]{{\buildrel{#1}\over\rightarrow\;}}

\newtheorem{theorem}{Theorem}[section]

\newtheorem*{remark}{Remark}

\begin{document}

\title{A Markov chain approach to
renormalization group transformations}

\author{Mei Yin}
\address{Department of Mathematics, University of Texas, Austin,
TX, 78712, USA} \email{myin@math.utexas.edu}

\dedicatory{\rm December 6, 2012}

\begin{abstract}
We aim at an explicit characterization of the renormalized
Hamiltonian after decimation transformation of a one-dimensional
Ising-type Hamiltonian with a nearest-neighbor interaction and a
magnetic field term. To facilitate a deeper understanding of the
decimation effect, we translate the renormalization flow on the
Ising Hamiltonian into a flow on the associated Markov chains
through the Markov-Gibbs equivalence. Two different methods are
used to verify the well-known conjecture that the eigenvalues of
the linearization of this renormalization transformation about the
fixed point bear important information about all six of the
critical exponents. This illustrates the universality property of
the renormalization group map in this case.

\vskip.1truein

\end{abstract}

\maketitle

\section{Introduction}
\label{intro} The discovery of the equivalence of Markov random
fields and Gibbs random fields was a major breakthrough in the
interchange of ideas between probability and physics. A Markov
random field is a natural generalization of the familiar concept
of a Markov chain, which is a collection of random variables with
the property that, given the present, the future is
(conditionally) independent of the past. If we look at the chain
itself as a very simple graph and ignore the directionality
implied by ``time'', then a Markov chain may alternatively be
viewed as a chain graph of stochastic variables, where each
variable is independent of all other variables (both future and
past) given its two neighbors. A Markov random field is the same
thing, only that rather than a chain graph, we allow the
relationship between the variables to be defined by any graph
structure, and each variable is independent of all the others
given its neighbors in the graph. A Gibbs random field, on the
other hand, is formed by a set of random variables whose
configurations obey a Gibbs distribution, which is a probability
distribution that factorizes over all possible cliques, i.e.
complete subgraphs in the graph, and the factors are conveniently
referred to as ``clique potentials''. These two ways of defining a
random configuration are apparently quite different
\cite{Honerkamp}: A Markov random field is characterized by its
local property (the Markovianity) whereas a Gibbs random field is
specified by its global property (the Gibbs distribution).

The rigorous study of the relationship between these two seemingly
unrelated fields was initiated by Dobrushin \cite{Dobrushin} in
the context of statistical physics, who considered the questions
of existence and uniqueness of a random field subject to a
Markovian conditional distribution. Further investigations quickly
ensued. Averintsev \cite{Averintsev} and Spitzer \cite{Spitzer}
independently proved that the class of two-state Markov chains is
identical to the class of Gibbs ensembles on the simple cubic
lattice. Hammersley and Clifford \cite{HC} showed that the same
equivalence holds between a multi-state Markov field and a
generalized Gibbs ensemble over an arbitrary finite graph. The
celebrated Hammersley-Clifford theorem states that each Markov
field with a system of neighbors and the associated system of
cliques is also a Gibbs field with the same system of cliques, and
vice versa, each Gibbs field is also a Markov field with the
corresponding system of neighbors. This implies that the joint
probability and the conditional probability can specify each
other, and serves as a theoretical basis for many modeling
applications, where the global characteristic is captured and
represented through a set of tractable local characteristics. The
original method of proof, however, did not have great intuitive
appeal, and many alternative proofs of this theorem were
developed. Sherman \cite{Sherman} verified the equivalence of
Markov fields and Gibbs ensembles under more relaxed conditions by
the repeated use of the inclusion-exclusion principle. Preston
\cite{Preston} adopted a direct approach to the two-state problem
and presented an explicit formula for the pair potential. Grimmett
\cite{Grimmett} showed that the equivalence of structure follows
immediately from an application of the M\"{o}bius inversion
theorem. A final improvement was done by Besag \cite{Besag}, who
applied methods of statistical analysis and gave a much simpler,
analytical proof of the general result.

The nearest-neighbour Ising model in one dimension is commonly
used to demonstrate the powerful Markov-Gibbs equivalence. Though
an ordered phase only emerges at zero temperature, this classic
model is physically important in that it has a fixed point (the
so-called ``zero temperature phase transition'') where the
critical exponents may be sensibly defined as in higher
dimensions. There is the astonishing empirical fact that these
critical exponents depend only on overall features of the system,
and are related to eigenvalues of the linearized renormalization
group map near the fixed point \cite{Wilson}. This universality
conjecture has generated continued interest in the scientific
community, and various approaches to the renormalization effect on
the one-dimensional Ising model have been explored \cite{Nelson,
Nauenberg}

Consider a one-dimensional Ising model with $N$ spins
$\sigma_i=\pm 1$, labelled successively $i=0,...,N-1$. We take the
system size $N$ to be very large (strictly speaking, infinite).
The Gibbs field of this model is described by a Hamiltonian $H$,
consisting of a nearest-neighbor interaction $J$ and a magnetic
field term $m$:
\begin{equation}
\label{H} H=-\left(J\sum_{i=0}^{N-1}
\sigma_i\sigma_{i+1}+m\sum_{i=0}^{N-1} \sigma_i\right),
\end{equation}
where periodic boundary condition is imposed so that
$\sigma_N=\sigma_0$, a standard setup to ensure that $H$ is
translation-invariant. We focus on a specific renormalization
group transformation, namely decimation transformation with
blocking factor $b$. To avoid unnecessary technicalities, we
assume that $b$ divides $N$. The decimation procedure is
straightforward: Fix the spins $\sigma_{bi}$ for $i=0,...,N/b-1$,
and integrate out the remaining ones. This will generate a
renormalized Gibbs field with a Hamiltonian $H'$ having the same
form as the original Hamiltonian $H$, but containing a
nearest-neighbor interaction $J'$ and a magnetic field term $m'$:
\begin{equation}
\label{H'} H'=-\left(J'\sum_{i=0}^{N/b-1}
\sigma_{bi}\sigma_{b(i+1)}+m'\sum_{i=0}^{N/b-1}
\sigma_{bi}\right).
\end{equation}
The renormalized spin coefficients $(J', m')$ and the original
spin coefficients $(J, m)$ are related by the decimation map:
\begin{equation*}
\exp\left(C+J'\sigma_0\sigma_b+\frac{m'}{2}\left(\sigma_0+\sigma_{b}\right)\right)
\end{equation*}
\begin{equation}
\label{RG}
=\sum_{\sigma_1,...,\sigma_{b-1}}\exp\left(J\sum_{i=0}^{b-1}\sigma_i\sigma_{i+1}+\frac{m}{2}\sum_{i=0}^{b-1}
\left(\sigma_i+\sigma_{i+1}\right)\right),
\end{equation}
where $C$ is a normalization constant. Notice that to avoid double
counting, we have assigned a ``half'' of the magnetic field $m$
($m'$) to each spin.

We would like to obtain an explicit characterization of the
renormalized model, but as the blocking factor $b$ gets large,
solving for $(J', m')$ directly from (\ref{RG}) becomes very
difficult. We thus take an alternative approach and investigate
the decimation effect on the associated Markov chains. As there is
no finite phase transition in one dimension, we follow the common
practice and measure the nearest-neighbor interaction strength $J$
($J'$) by the Boltzmann factor $k=e^{-2J}$ ($k'=e^{-2J'}$)
instead. An explicit solution for $(k', m')$ then follows from the
Markov-Gibbs equivalence (Hammersley-Clifford theorem). The
diagram below illustrates these ideas:

\begin{equation}
\label{diagram}
\begin{array}{ccc}
  \text{original Hamiltonian} & \overlim{(\text{I})} & \text{renormalized Hamiltonian} \\
  \downarrow_{(\text{II})} &  & \uparrow_{(\text{IV})} \\
  \text{original Markov chain} & \overlim{(\text{III})} & \text{renormalized Markov chain} \\
\end{array}
\end{equation}
where:
\begin{itemize}
\item (II) and (IV) indicate the Markov-Gibbs equivalence
(Hammersley-Clifford theorem).

\item (I) is the decimation map on the Ising Hamiltonian (cf.
(\ref{RG})).

\item (III) is the decimation map on the associated Markov chains
(to be examined).
\end{itemize}

A key tenet of the renormalization group is its explanation of
universality \cite{vanEnter}. Thus we would also like to verify
the widely-believed universality conjecture in this special case,
which states that the linearization of the decimation
transformation with blocking factor $b$ about the two-dimensional
fixed point ($k=m=0$) has two real eigenvalues $b^{y_T}$ and
$b^{y_H}$, where $y_T=y_H=1$. Suppose we start with a Hamiltonian
that is close to critical. The decimation map will first drive it
towards the fixed point for a large number of iterations, but
eventually will drive it away. The singular behavior of the model
arises from iterating the map infinitely many times, and the
critical properties are determined by how much time the
Hamiltonian spends near the fixed point, when its behavior is
governed by the linearization. In fact, it is observed that there
are exact non-trivial relations between the six critical exponents
(specific heat $\alpha$, spontaneous magnetization $\beta$,
magnetic susceptibility $\gamma$, response to magnetic field at
zero temperature $\delta$, correlation length $\nu$, and
correlation function at zero temperature $\eta$) and the two
eigenvalues (more precisely $y_T$ and $y_H$) of the linearization:
\begin{equation}
\alpha=2-\frac{d}{y_T}=1, \hspace{1cm} \beta=\frac{d-y_H}{y_T}=0,
\hspace{1cm} \gamma=\frac{2y_H-d}{y_T}=1,
\end{equation}
\begin{equation}
\delta=\frac{y_H}{d-y_H}=\infty, \hspace{1cm} \nu=\frac{1}{y_T}=1,
\hspace{1cm} \eta=d+2-2y_H=1.
\end{equation}
(More discussions may be found in \cite{Nelson} and
\cite{Baxter}.)

\begin{theorem}[Universality Conjecture]
\label{prop} At the fixed point ($k=m=0$), the Jacobian matrix of
the renormalized spin coefficients $(k', m')$ with respect to the
original spin coefficients $(k, m)$ is given by
\begin{eqnarray}
\label{J}
\textsf{Jac}=\left(%
\begin{array}{cc}
  \frac{\partial k'}{\partial k} & \frac{\partial k'}{\partial m} \\
  \frac{\partial m'}{\partial k} & \frac{\partial m'}{\partial m} \\
\end{array}%
\right)=\left(%
\begin{array}{cc}
  b & 0 \\
  0 & b \\
\end{array}%
\right).
\end{eqnarray}
\end{theorem}

The rest of this paper is organized as follows. In Section
\ref{renormalization} we verify the universality conjecture by
analyzing the decimation map on the Ising Hamiltonian directly
(First Proof of Theorem \ref{prop}). In Section \ref{markov} the
statistical physics model is transformed into a probability model
through the Markov-Gibbs equivalence (Theorems \ref{model} and
\ref{rmodel}). We investigate the decimation effect on the
associated Markov chains and give an explicit characterization of
the renormalized Hamiltonian (Theorem \ref{char}). An alternative
proof of the universality conjecture from this point of view is
also provided (Second Proof of Theorem \ref{prop}). Finally,
Section \ref{conclusion} is devoted to concluding remarks.

\section{Renormalization group approach}
\label{renormalization} In this section we will examine the
renormalization group equation (\ref{RG}) directly. Although it is
difficult to find an explicit solution to (\ref{RG}) for a large
blocking factor $b$, the Jacobian matrix of partial derivatives
(\ref{J}) may be computed via implicit differentiation.

\vskip.1truein

\noindent \textit{First Proof of Theorem \ref{prop}.} The
decimation map (\ref{RG}) consists of 4 equations.

\noindent 1. Corresponding to $\sigma_0=\sigma_b=1$:
\begin{equation*}
\exp\left(C+J'+m'\right)
\end{equation*}
\begin{equation}
\label{1}
=\sum_{\sigma_1,...,\sigma_{b-1}}\exp\left(J\left(\sigma_1+\sum_{i=1}^{b-2}\sigma_i\sigma_{i+1}+\sigma_{b-1}\right)+m\sum_{i=1}^{b-1}
\sigma_i+m\right).
\end{equation}
2. Corresponding to $\sigma_0=\sigma_b=-1$:
\begin{equation*}
\exp\left(C+J'-m'\right)
\end{equation*}
\begin{equation}
\label{2}
=\sum_{\sigma_1,...,\sigma_{b-1}}\exp\left(J\left(-\sigma_1+\sum_{i=1}^{b-2}\sigma_i\sigma_{i+1}-\sigma_{b-1}\right)+m\sum_{i=1}^{b-1}
\sigma_i-m\right).
\end{equation}
3. Corresponding to $\sigma_0=1,\sigma_b=-1$:
\begin{equation*}
\exp\left(C-J'\right)
\end{equation*}
\begin{equation}
\label{3}
=\sum_{\sigma_1,...,\sigma_{b-1}}\exp\left(J\left(\sigma_1+\sum_{i=1}^{b-2}\sigma_i\sigma_{i+1}-\sigma_{b-1}\right)+m\sum_{i=1}^{b-1}
\sigma_i\right).
\end{equation}
4. Corresponding to $\sigma_0=-1,\sigma_b=1$:
\begin{equation*}
\exp\left(C-J'\right)
\end{equation*}
\begin{equation}
\label{4}
=\sum_{\sigma_1,...,\sigma_{b-1}}\exp\left(J\left(-\sigma_1+\sum_{i=1}^{b-2}\sigma_i\sigma_{i+1}+\sigma_{b-1}\right)+m\sum_{i=1}^{b-1}
\sigma_i\right).
\end{equation}

Due to symmetry, (\ref{3}) and (\ref{4}) are equivalent. We may
therefore assume that (\ref{RG}) breaks down into 3 equations:
(\ref{1}), (\ref{2}), and (\ref{3}). To compute the Jacobian
matrix of the decimation transformation at the fixed point
($k=m=0$), we perform implicit differentiation on these equations
at ($J=\infty$, $m=0$). As an example, we differentiate both sides
of (\ref{1}) with respect to $m$, which gives
\begin{equation*}
\frac{\partial C}{\partial m}+\frac{\partial J'}{\partial
m}+\frac{\partial m'}{\partial m}-1
\end{equation*}
\begin{eqnarray}
=\frac{\sum
\limits_{\sigma_1,...,\sigma_{b-1}}\exp\left(J\left(\sigma_1+\sum
\limits_{i=1}^{b-2}\sigma_i\sigma_{i+1} +\sigma_{b-1}\right)+m
\sum \limits_{i=1}^{b-1}\sigma_i\right) \left(\sum
\limits_{i=1}^{b-1}\sigma_i\right)}{\sum
\limits_{\sigma_1,...,\sigma_{b-1}} \exp\left(J\left(\sigma_1+\sum
\limits_{i=1}^{b-2}\sigma_i\sigma_{i+1}+\sigma_{b-1}\right)+m \sum
\limits_{i=1}^{b-1}\sigma_i\right)}.
\end{eqnarray}
Because lower order terms become insignificant at ``$J=\infty$'',
it suffices to keep track of the ``dominating terms''. We have
\begin{equation*}
\frac{\partial C}{\partial m}+\frac{\partial J'}{\partial
m}+\frac{\partial m'}{\partial m}-1
\end{equation*}
\begin{equation}
=\frac{s\left(\sum_{i=1}^{b-1}\sigma_i|\text{max}\left(\sigma_1+\sum_{i=1}^{b-2}\sigma_i\sigma_{i+1}+\sigma_{b-1}\right)\right)}
{n\left(\sigma_1+\sum_{i=1}^{b-2}\sigma_i\sigma_{i+1}+\sigma_{b-1}\right)},
\end{equation}
where $n(f)$ counts the number of $\sigma$ configurations that
maximize $f(\sigma)$, and $s(g|\text{max}(f))$ (which we will
abbreviate by $s(g)$) is the sum of $g(\sigma)$ over the
maximizers of $f(\sigma)$. Repeating this ``dominating'' procedure
provides us with 6 independent equations for the partial
derivatives:
\begin{equation}
\label{first} \frac{\partial C}{\partial J}+\frac{\partial
J'}{\partial J}+\frac{\partial m'}{\partial
J}=\frac{s(\sigma_1+\sum_{i=1}^{b-2}\sigma_i\sigma_{i+1}+\sigma_{b-1})}{n(\sigma_1+\sum_{i=1}^{b-2}\sigma_i\sigma_{i+1}+\sigma_{b-1})},
\end{equation}
\begin{equation}
\frac{\partial C}{\partial m}+\frac{\partial J'}{\partial
m}+\frac{\partial m'}{\partial
m}=\frac{s(\sum_{i=1}^{b-1}\sigma_i)}{n(\sigma_1+\sum_{i=1}^{b-2}\sigma_i\sigma_{i+1}+\sigma_{b-1})}+1,
\end{equation}
\begin{equation}
\frac{\partial C}{\partial J}+\frac{\partial J'}{\partial
J}-\frac{\partial m'}{\partial
J}=\frac{s(-\sigma_1+\sum_{i=1}^{b-2}\sigma_i\sigma_{i+1}-\sigma_{b-1})}{n(-\sigma_1+\sum_{i=1}^{b-2}\sigma_i\sigma_{i+1}-\sigma_{b-1})},
\end{equation}
\begin{equation}
\frac{\partial C}{\partial m}+\frac{\partial J'}{\partial
m}-\frac{\partial m'}{\partial m}=\frac{s(\sum_{i=1}^{b-1}
\sigma_i)}{n(-\sigma_1+\sum_{i=1}^{b-2}\sigma_i\sigma_{i+1}-\sigma_{b-1})}-1,
\end{equation}
\begin{equation}
\frac{\partial C}{\partial J}-\frac{\partial J'}{\partial
J}=\frac{s(\sigma_1+\sum_{i=1}^{b-2}\sigma_i\sigma_{i+1}-\sigma_{b-1})}{n(\sigma_1+\sum_{i=1}^{b-2}\sigma_i\sigma_{i+1}-\sigma_{b-1})},
\end{equation}
\begin{equation}
\label{last} \frac{\partial C}{\partial m}-\frac{\partial
J'}{\partial m}=\frac{s(\sum_{i=1}^{b-1}
\sigma_i)}{n(\sigma_1+\sum_{i=1}^{b-2}\sigma_i\sigma_{i+1}-\sigma_{b-1})}.
\end{equation}

It is quite clear that
$\max(\sigma_1+\sum_{i=1}^{b-2}\sigma_i\sigma_{i+1}+\sigma_{b-1})=b$
is achieved only when $\sigma_1=\cdots=\sigma_{b-1}=1$, and that
$\max(-\sigma_1+\sum_{i=1}^{b-2}\sigma_i\sigma_{i+1}-\sigma_{b-1})=b$
is achieved only when $\sigma_1=\cdots=\sigma_{b-1}=-1$. The
harder task it to determine when
$\max(\sigma_1+\sum_{i=1}^{b-2}\sigma_i\sigma_{i+1}-\sigma_{b-1})$
is obtained. Because of the single ``$-$'' sign, it is not
possible for all the $b$ terms in this sum ($\sigma_1,
\sigma_1\sigma_2,...,\sigma_{b-2}\sigma_{b-1}, -\sigma_{b-1}$) to
be $1$ as in the previous two cases. An ideal maximizer should
have $b-1$ terms with value $1$ and only one term with value $-1$.
We claim that each one of the $b$ locations of $-1$ corresponds to
exactly one $\sigma$ configuration: Suppose the $i$th term has
value $-1$ ($\sigma_1=-1$ for $i=1$, $\sigma_{i-1}\sigma_i=-1$ for
$2\leq i\leq b-1$, or $\sigma_{b-1}=1$ for $i=b$), then we must
have $\sigma_1=\cdots=\sigma_{i-1}=1$ and
$\sigma_i=\cdots=\sigma_{b-1}=-1$. (\ref{first})---(\ref{last})
are thus simplified:
\begin{equation}
\label{one} \frac{\partial C}{\partial J}+\frac{\partial
J'}{\partial J}+\frac{\partial m'}{\partial J}=b, \hspace{1cm}
\frac{\partial C}{\partial m}+\frac{\partial J'}{\partial
m}+\frac{\partial m'}{\partial m}=b,
\end{equation}
\begin{equation}
\frac{\partial C}{\partial J}+\frac{\partial J'}{\partial
J}-\frac{\partial m'}{\partial J}=b, \hspace{1cm} \frac{\partial
C}{\partial m}+\frac{\partial J'}{\partial m}-\frac{\partial
m'}{\partial m}=-b,
\end{equation}
\begin{equation}
\label{three} \frac{\partial C}{\partial J}-\frac{\partial
J'}{\partial J}=b-2, \hspace{1cm} \frac{\partial C}{\partial
m}-\frac{\partial J'}{\partial m}=0.
\end{equation}

Solving (\ref{one})---(\ref{three}) yields
\begin{equation}
\frac{\partial J'}{\partial J}=1, \hspace{1cm} \frac{\partial
J'}{\partial m}=\frac{\partial m'}{\partial J}=0, \hspace{1cm}
\frac{\partial m'}{\partial m}=b,
\end{equation}
which further implies that the Jacobian matrix $\textsf{Jac}$
(\ref{J}) is diagonal, i.e., $\frac{\partial k'}{\partial
m}=\frac{\partial m'}{\partial k}=0$. To complete the proof of the
universality conjecture, it remains to verify that $\frac{\partial
k'}{\partial k}=b$. We perform the ``dominating'' procedure as
before. For notational convenience, we temporarily denote
$\sigma_1+\sum_{i=1}^{b-2}\sigma_i\sigma_{i+1}-\sigma_{b-1}$ by
$f(\sigma)$, and
$\sigma_1+\sum_{i=1}^{b-2}\sigma_i\sigma_{i+1}+\sigma_{b-1}$ by
$g(\sigma)$. Dividing (\ref{3}) by (\ref{1}) at the fixed point
$(J=\infty, m=0)$, we have
\begin{equation*}
k'=\exp(-2J')=\frac{n(f)\exp\left(J\cdot\text{max}(f)\right)}{n(g)\exp\left(J\cdot\text{max}(g)\right)}
\end{equation*}
\begin{equation}
=\frac{b\exp\left((b-2)J\right)}{\exp(bJ)}=b\exp(-2J)=bk.
\end{equation} \qed

\section{Markov chain approach}
\label{markov} In this section we will transform the statistical
physics model into a probability model and investigate the
decimation effect on the associated Markov chains. This is a
special case of Hammersley-Clifford theorem where the exact
correspondence between the Markov field and the Gibbs field may be
worked out explicitly. The idea is to regard the Ising system as a
two-state Markov chain with transition probability matrix
\begin{equation}
P=\left(%
\begin{array}{cc}
  1-p & p \\
  q & 1-q \\
\end{array}%
\right),
\end{equation}
where
\begin{equation}
p=\PR(\sigma_1=1|\sigma_0=-1),
\end{equation}
and
\begin{equation}
q=\PR(\sigma_1=-1|\sigma_0=1).
\end{equation}

\begin{theorem}[Hammersley-Clifford]
\label{model} The Ising Hamiltonian $H$ (\ref{H}) is fully
characterized by the transition probabilities $p$ and $q$.
\end{theorem}

\begin{remark}
The transition probabilities $p$ and $q$ and the spin coefficients
$k$ and $m$ are related by (\ref{p}), (\ref{q}), (\ref{JJ}), and
(\ref{m}). The spin coefficients fixed point ($k=m=0$) thus
corresponds to the transition probabilities fixed point ($p=q=0$).
\end{remark}

\begin{proof}
Baxter \cite{Baxter} showed that the mean and covariance of the
Ising spins in the infinite-volume limit are functions of the spin
coefficients $k$ and $m$:
\begin{equation}
\label{E} \E\sigma_0=\frac{\sinh m}{\sqrt{\sinh^2 m+k^2}},
\end{equation}
\begin{equation}
\label{cov} \Cov(\sigma_0, \sigma_1)=\frac{k^2}{\sinh^2
m+k^2}\frac{\cosh m - \sqrt{\sinh^2 m+k^2}}{\cosh m +
\sqrt{\sinh^2 m+k^2}}.
\end{equation}
Through the Markov-Gibbs equivalence, we show that (\ref{E}) and
(\ref{cov}) may alternatively be viewed as functions of the
transition probabilities $p$ and $q$. Recall that the Markov chain
has a stationary distribution:
\begin{equation}
\PR(\sigma_0=-1)=\frac{q}{p+q}, \hspace{1cm}
\PR(\sigma_0=1)=\frac{p}{p+q},
\end{equation}
which readily displays the dependence of the mean on the
transition probabilities,
\begin{equation}
\E\sigma_0=\E\sigma_1=\frac{p-q}{p+q}.
\end{equation}
To obtain an analogous expression for the covariance, we resort to
the tower property of conditional expectation,
\begin{equation*}
\E\sigma_0\sigma_1=\E(\sigma_0\E(\sigma_1|\sigma_0))
\end{equation*}
\begin{equation*}
=\PR(\sigma_1=1|\sigma_0=1)\PR(\sigma_0=1)-\PR(\sigma_1=-1|\sigma_0=1)\PR(\sigma_0=1)
\end{equation*}
\begin{equation*}
-\PR(\sigma_1=1|\sigma_0=-1)\PR(\sigma_0=-1)+\PR(\sigma_1=-1|\sigma_0=-1)\PR(\sigma_0=-1)
\end{equation*}
\begin{equation}
=\frac{(p-q)^2}{(p+q)^2}+(1-p-q)\frac{4pq}{(p+q)^2},
\end{equation}
which then gives
\begin{equation}
\Cov(\sigma_0,
\sigma_1)=\E\sigma_0\sigma_1-\E\sigma_0\E\sigma_1=(1-p-q)\frac{4pq}{(p+q)^2}.
\end{equation}
The two characterizations of the Ising Hamiltonian $H$ (\ref{H})
are thus connected by:
\begin{equation}
\label{char1} \frac{p-q}{p+q}=\frac{\sinh m}{\sqrt{\sinh^2
m+k^2}},
\end{equation}
\begin{equation}
\label{char2} (1-p-q)\frac{4pq}{(p+q)^2}=\frac{k^2}{\sinh^2
m+k^2}\frac{\cosh m - \sqrt{\sinh^2 m+k^2}}{\cosh m +
\sqrt{\sinh^2 m+k^2}}.
\end{equation}

It is not hard to derive an explicit expression of $p$ and $q$ in
terms of $k$ and $m$ from (\ref{char1}) and (\ref{char2}):
\begin{equation}
\label{p} p=\frac{\sqrt{\sinh^2 m+k^2}+\sinh m}{\cosh
m+\sqrt{\sinh^2 m+k^2}},
\end{equation}
\begin{equation}
\label{q} q=\frac{\sqrt{\sinh^2 m+k^2}-\sinh m}{\cosh
m+\sqrt{\sinh^2 m+k^2}}.
\end{equation}
The reverse direction, however, requires more work. For
computational convenience, we make a change of variables, $A=\sinh
m$, $B=\sinh^2 m+k^2$. Then (\ref{p}) and (\ref{q}) become
\begin{equation}
\label{1st} \sqrt{B}+A=p\sqrt{A^2+1}+p\sqrt{B},
\end{equation}
\begin{equation}
\label{2nd} \sqrt{B}-A=q\sqrt{A^2+1}+q\sqrt{B}.
\end{equation}
Dividing (\ref{2nd}) into (\ref{1st}), we have
\begin{equation}
\frac{A-p\sqrt{A^2+1}}{-A-q\sqrt{A^2+1}}=\frac{p-1}{q-1}.
\end{equation}
This is an equation for $A$ only, and an explicit expression of
$k$ and $m$ in terms of $p$ and $q$ follows easily:
\begin{equation}
\label{JJ} k=\sqrt{\frac{pq}{(1-p)(1-q)}},
\end{equation}
\begin{equation}
\label{m} m=\frac{1}{2}\log\left(\frac{1-q}{1-p}\right).
\end{equation}
\end{proof}

\begin{theorem}
\label{rmodel} The renormalized Ising Hamiltonian $H'$ (\ref{H'})
is fully characterized by the renormalized transition
probabilities $p'$ and $q'$, where
\begin{equation}
p'=\PR(\sigma_b=1|\sigma_0=-1),
\end{equation}
and
\begin{equation}
q'=\PR(\sigma_b=-1|\sigma_0=1).
\end{equation}
\end{theorem}

\begin{remark}
The renormalized transition probabilities $p'$ and $q'$ and the
renormalized spin coefficients $k'$ and $m'$ are similarly related
as in (\ref{p}), (\ref{q}), (\ref{JJ}), and (\ref{m}).
\end{remark}

\begin{proof}
This follows from Theorem \ref{model} once we realize that site
$0$ and site $b$ are nearest neighbors after decimation
transformation with blocking factor $b$. The $b$-step transition
probability matrix $P^b$ represents the decimation map on the
associated Markov chains, and is given by
\begin{equation*}
P^b=\left(\left(%
\begin{array}{cc}
  1 & p \\
  1 & -q \\
\end{array}%
\right)\left(%
\begin{array}{cc}
  1 & 0 \\
  0 & 1-p-q \\
\end{array}%
\right)\left(%
\begin{array}{cc}
  \frac{q}{p+q} & \frac{p}{p+q} \\
  \frac{1}{p+q} & \frac{-1}{p+q} \\
\end{array}%
\right)\right)^b
\end{equation*}
\begin{equation}
=\left(%
\begin{array}{c}
  1 \\
  1 \\
\end{array}%
\right)\left(%
\begin{array}{cc}
  \frac{q}{p+q} & \frac{p}{p+q} \\
\end{array}%
\right)+(1-p-q)^b\left(%
\begin{array}{c}
  p \\
  -q \\
\end{array}%
\right)\left(%
\begin{array}{cc}
  \frac{1}{p+q} & \frac{-1}{p+q} \\
\end{array}%
\right),
\end{equation}
where the first equality is simply the spectral decomposition of
the matrix $P$. This then implies that
\begin{equation}
\label{p'} p'=\frac{p}{p+q}(1-(1-p-q)^b),
\end{equation}
and
\begin{equation}
\label{q'} q'=\frac{q}{p+q}(1-(1-p-q)^b).
\end{equation}
\end{proof}

\begin{theorem}
\label{char} The decimation map (\ref{RG}) identifies the
connection between the renormalized Hamiltonian $H'$ (\ref{H'})
and the original Hamiltonian $H$ (\ref{H}).
\end{theorem}

\begin{proof}
We follow (II), (III), and (IV) as shown in (\ref{diagram}). The
original Ising model is described by a Hamiltonian $H$ with spin
coefficients $k$ and $m$. (II) indicates the alternative view of
this system as a two-state Markov chain with transition
probabilities $p$ and $q$ (cf. (\ref{p}) and (\ref{q})). (III)
then transforms this Markov chain into a renormalized Markov chain
with renormalized transition probabilities $p'$ and $q'$ (cf.
(\ref{p'}) and (\ref{q'})). Finally, (IV) recovers the
renormalized spin coefficients $k'$ and $m'$ of the renormalized
Hamiltonian $H'$ (cf. (\ref{JJ}) and (\ref{m})).
\begin{equation}
\label{rel}
\begin{array}{ccccccc}
  (k, m) & \overlim{(\text{II})} & (p, q) & \overlim{(\text{III})} & (p', q') & \overlim{(\text{IV})} & (k', m')\\
\end{array}
\end{equation}
\end{proof}

\noindent \textit{Second Proof of Theorem \ref{prop}.} Theorem
\ref{char} establishes an explicit expression of the renormalized
spin coefficients $k'$ and $m'$ in terms of the original spin
coefficients $k$ and $m$ (cf. (\ref{rel})). To evaluate the
Jacobian matrix $\textsf{Jac}$ (\ref{J}) at the fixed point
$(k=m=0)$, we start by considering $\frac{\partial k'}{\partial
k}$ and $\frac{\partial m'}{\partial k}$ with $m$ held fixed at
zero. By (II), on the $m=0$ curve,
\begin{equation}
\label{newp} p=q=\frac{k}{1+k}.
\end{equation}
(III) then gives
\begin{equation}
\label{newp'} p'=q'=\frac{1}{2}\left(1-(1-2p)^b\right),
\end{equation}
which further implies, by (IV), that
\begin{equation}
\label{newk'} k'=\frac{p'}{1-p'},
\end{equation}
\begin{equation}
\label{friday}m'=0.
\end{equation}
We conclude that $\frac{\partial m'}{\partial k}=0$ from
(\ref{friday}), and by applying the chain rule to (\ref{newp}),
(\ref{newp'}), and (\ref{newk'}), that $\frac{\partial
k'}{\partial k}=b$.

We proceed with the calculations for $\frac{\partial m'}{\partial
m}$ and $\frac{\partial k'}{\partial m}$ with $k$ held fixed at
zero. By (II), on the $k=0$ curve, either $p$ or $q$ is zero,
depending on the sign of $m$. Without loss of generality, assume
$m\geq 0$. In this case,
\begin{equation}
\label{anotherp} p=\frac{2\sinh m}{\cosh m+\sinh m},
\end{equation}
\begin{equation}
q=0.
\end{equation}
(III) then gives
\begin{equation}
\label{anotherp'} p'=1-(1-p)^b,
\end{equation}
\begin{equation}
q'=0,
\end{equation}
which further implies, by (IV), that
\begin{equation}
\label{saturday} k'=0,
\end{equation}
\begin{equation}
\label{anotherm'} m'=-\frac{1}{2}\log(1-p').
\end{equation}
We conclude that $\frac{\partial k'}{\partial m}=0$ from
(\ref{saturday}), and by applying the chain rule to
(\ref{anotherp}), (\ref{anotherp'}), and (\ref{anotherm'}), that
$\frac{\partial m'}{\partial m}=b$. \qed

\section{Concluding remarks}
\label{conclusion} This paper aims at an explicit characterization
of the renormalized Hamiltonian after decimation transformation of
a one-dimensional Ising-type Hamiltonian with a nearest-neighbor
interaction and a magnetic field term. We transform the
statistical physics model into a probability model through the
Markov-Gibbs equivalence and analyze the decimation effect on the
associated Markov chains. As the Ising model is a prototype for a
wide variety of spin models, it is expected that the exploitation
of Markov-Gibbs equivalence in this special case will shed light
on the application of renormalization group ideas in a more
general setting. Two different proofs of the universality
conjecture are presented, one based directly upon the
renormalization group equation, and the other from the Markov
chain point of view. Although the first proof does not employ
advanced mathematical methods, it provides a new perspective on
the renormalization flow. For example, it has been verified,
following similar ideas, that one-dimensional $q$-state Potts
model ($q\geq 2$) exhibits the same eigenvalue statistics $y_T$
and $y_H$, independent of the number of states $q$ and the
blocking factor $b$. (The percolation limit $q\to 1$, however,
remains open, and is believed to display different critical
features.) The second proof uses ideas from Markov chains, and is
expected to work with higher-dimensional $q$-state Potts models as
well, where the covariant matrices may be expressed in terms of
the random cluster representation of Fortuin and Kasteleyn
\cite{FK}. As the number of dimensions $d$ and the number of
states $q$ get large, it will be harder to write down exact
formulas for the transition probabilities in the covariant
matrices, but the Metropolis and Glauber algorithms should provide
a reasonable approximation scheme. Since Markov chains may take
both discrete and continuous values, an advantage of exploring
this second perspective is that we can also consider decimation
with spin scaling applied to continuous spin systems, not just
discrete systems like Potts and Ising models, and hence avoid the
lack of spin rescaling, a common problem encountered in a ``pure''
decimation. In summary, we hope this rigorous investigation will
provide insight into the intrinsic structure of the
renormalization group transformation and help us better understand
the nature of universality.

\section*{Acknowledgements}
The author owes deep gratitude to her PhD advisor Bill Faris for
his continued help and support. She appreciated the opportunity to
talk about an early version of this work in the 2010 Arizona
School of Analysis with Applications, organized by Bob Sims and
Daniel Ueltschi. This research was supported in part by the R. H.
Bing Fellowship at University of Texas at Austin.

\end{document}